\documentclass{llncs}

\usepackage{amsfonts}
\usepackage{amsmath}
\usepackage{float}
\usepackage{todonotes}
\usepackage{subfigure}
\usepackage[title]{appendix}
\let\cleardoublepage\clearpage

\floatstyle{ruled}
\newfloat{algorithm}{th}{lop}
\floatname{algorithm}{Algorithm Sketch}

\spnewtheorem*{proofsketch}{Proof Sketch}{\itshape}{}

\begin{document}

\title{Applications and Implications of a General Framework for Self-Stabilizing Overlay Networks}
\author{Andrew Berns}
\institute{Department of Computer Science\\University of Northern Iowa\\Cedar Falls, IA USA\\
\email{andrew.berns@uni.edu}}



\maketitle

\begin{abstract}
  From data centers to IoT devices to Internet-based applications, overlay networks have become an important part of modern computing.  Many of these overlay networks operate in fragile environments where processes are susceptible to faults which may perturb a node's state and the network topology.  Self-stabilizing overlay networks have been proposed as one way to manage these faults, promising to build or restore a particular topology from any initial configuration or after the occurrence of any transient fault.  To date there have been several self-stabilizing protocols designed for overlay networks.  These protocols, however, are either focused on a single specific topology, or provide very inefficient solutions for a general set of overlay networks.

  In this paper, we analyze an existing algorithm and show it can be used as a general framework for building many other self-stabilizing overlay networks.  Our analysis for time and space complexity depends upon several properties of the target topology itself, providing insight into how topology selection impacts the complexity of convergence.  We then demonstrate the application of this framework by analyzing the complexity for several existing topologies.  Next, using insights gained from our analysis, we present a new topology designed to provide efficient performance during convergence with the general framework.  Our process demonstrates how the implications of our analysis help isolate the factors of interest to allow a network designer to select an appropriate topology for the problem requirements.

\end{abstract}

\section{Introduction}
Distributed systems have become an ubiquitous part of modern computing, with systems continuing to grow in size and scope.  As these systems grow larger, the need for topologies that allow for efficient operations like search and routing increases.  To this end, many systems use \emph{overlay networks} to control the network topology.  In overlay networks, connections are made using logical links, each of which consists of zero or more physical links.  This use of logical links means program actions can add and delete edges in the network, allowing the system to maintain an arbitrary logical topology even when the physical topology may be fixed.

Many of these large systems operate in environments where faults are commonplace.  Servers may crash, communication links may be damaged, and processes may join or leave the system frequently.  This reality has increased the demand for fault tolerant overlay networks.  One particularly strong type of fault tolerance is \emph{self-stabilization}, where a legal configuration is guaranteed to be reached after any transient fault.  For overlay networks, this means a correct topology can always be built when starting from any configuration provided the network is not disconnected.

\subsection{Problem Overview}
Our current work focuses on self-stabilizing overlay networks.  A self-stabilizing overlay network guarantees that program actions will build a legal topology even when the system starts in \emph{any} weakly-connected topology.  Our interest, then, is in the design and analysis of algorithms that, when executed on an arbitrary initial weakly-connected topology, add and delete edges with program actions until a legal \emph{target topology} is reached.

Going further, we are interested in algorithms and analyses for general frameworks for overlay network creation.  To date, most work has focused on algorithms for a single topology, or has been inefficient in terms of time or space complexity.  Work focused on a specific topology is hard to generalize and derive insights from for expanding to other overlay network applications, while general frameworks with high time and space complexity may be too inefficient to be useful in practice.  Our interest is in general frameworks for overlay network creation that allow efficient algorithms to be built while still being general enough to provide insights into the stabilization of arbitrary topologies.

\subsection{Main Results and Significance}
In this paper, we build upon the work of Berns~\cite{berns_avatar_15} to present a general algorithm for creating self-stabilizing overlay network protocols for a variety of target topologies, and provide several examples of the application of our general algorithm, including with a new overlay network topology.  More specifically, our contributions are as follows:

\begin{itemize}
    \item We update the analysis of Berns~\cite{berns_avatar_15} to show how their algorithm can be extended into an algorithm for \emph{any} target topology.  Our updated analysis is the first to show a general framework for self-stabilizing overlay network creation that allows for efficient stabilization in both time and space.
    \item As part of this updated analysis, we introduce several measures of complexity that are properties of the target topology itself.  These measures are useful for two reasons.  First, they allow us to analyze the general algorithm easily for a variety of topologies.  Perhaps just as important, however, is that they provide valuable insight into how the selection of the target topology affects convergence in terms of both time and space.  This insight can be quite useful for not only selecting a target topology, but also in designing new overlay network topologies.
    \item We analyze several existing overlay network topologies, providing the necessary metrics for analysis in our general framework.  This analysis provides an example of how our framework can be applied, and also helps provide concrete insight into factors affecting convergence and demonstrate how diameter, degree, and robustness are balanced when designing self-stabilizing overlay networks.
    \item Using the insight gained from our earlier contribution, we define a new network topology which stabilizes with sublinear time and space complexity in our framework.  This demonstrates how the framework can provide network designers with guidance to help them build new topologies that can stabilize efficiently with our approach.  The design and analysis of new topologies targeted for efficient stabilization in this framework could be an interesting area of future study.
\end{itemize}

The key idea of our work is the extension of the algorithm of Berns~\cite{berns_avatar_15} to work with other topologies by defining the target topology and analyzing several relevant measures regarding this topology.  This definition allows the creation and analysis of many self-stabilizing overlay networks without having to design the algorithm from scratch.  Furthermore, our framework highlights the factors of the network that affect stabilization, allowing a designer to tune the topology to meet their needs.

\subsection{Related Work and Comparison}
The past few decades have seen a large body of work develop on overlay networks.  Early work focused on defining \emph{structured} overlay networks, where a single correct configuration existed for a particular set of nodes.  Examples of these networks include \textsc{Chord}~\cite{stoica_chord_01} and \textsc{Tapestry}~\cite{zhao_tapestry_2004}.  These works often did not consider fault tolerance, or considered a weak model with limited possible failures.

As work expanded in overlay networks, so did work in various types of fault tolerance.  One category of work considered \emph{self-healing} networks, where a particular network property could be maintained even during limited node deletions~\cite{hayes_forgiving-tree_2008,hayes_forgiving-graph_2009}.  Several examples of this work even used virtual nodes~\cite{trehan_virtual_12}, although they were not used to create a specific embedding as done in this current work.  Recently in \textsc{DConstructor}~\cite{gilbert_dconstructor_20}, the authors present a framework for building overlay networks.  \textsc{DConstructor} works by forming clusters and merging these clusters together.  However, as with the other examples, \textsc{DConstructor} is not self-stabilizing as it assumes all nodes begin in a single node cluster.  Said in another way, \textsc{DConstructor} assumes an arbitrary initial \emph{topology}, but not an arbitrary initial state.  This is also the same assumption in the work of G\"{o}tte et al.~\cite{gotte_time-optimal_podc21}, who presented an algorithm for transforming a constant-degree network into a tree in $\mathcal{O}(\log n)$ rounds.

Our work considers \emph{self-stabilizing overlay networks}, where the correct configuration is reached after an arbitrary number of transient faults that do not disconnect the network.  There are several examples of these as well.  The \textsc{Skip+} graph~\cite{jacob_skipplus_09} presents a self-stabilizing variant of the \textsc{Skip} graph~\cite{aspnes_skipgraph_03} with polylogarithmic convergence time, although the space requirements are linear for some configurations.  Other examples of self-stabilizing overlay networks include \textsc{Re-Chord}~\cite{kniesburges_rechord_11}, a \textsc{Chord} variant with virtual nodes that requires $\mathcal{O}(n \log n)$ rounds to converge, and \textsc{MultiSkipGraph}, a skip graph variant which maintained a property called \emph{monotonic searchability} during convergence~\cite{luo_multiskipgraph_2019}.

To date, most work has been focused on the convergence to a particular topology, with algorithms and analysis all targeted at these specific instances.  One exception to this is the \emph{Transitive Closure Framework}~\cite{berns_tcf_11}, which presents a general algorithm for creating any locally-checkable overlay network.  Like our work, they also identify a general measure of interest for stabilization time which they call the \emph{detector diameter}.  While this measure helps to bound the \emph{time} complexity of their algorithm, the space requirements are $\Theta(n)$ for any topology, limiting the applicability.

In \textsc{Avatar}~\cite{berns_avatar_15}, Berns presented both a locally-checkable definition of a network embedding for arbitrary topologies, and a self-stabilizing algorithm for building an embedded binary search tree with polylogarithmic time and space requirements.  The work only considered a single topology, however, and did not offer insight into measures for arbitrary topologies.  Our goal with this work is to build upon \textsc{Avatar} to address these issues.
\section{Preliminaries}
\label{section:prelims}
\subsection{Model of Computation}
We model our distributed system as an undirected graph $G = (V, E)$, with $n$ processes in $V$ communicating over the edges $E$.  Each node $u \in V$ has a unique identifier $u.id \in \mathbb{N}$, which is stored as immutable data in $u$.  Where clear from the context, we will use $u$ to represent the identifier of $u$.

Each node $u \in V$ has a \emph{local state} consisting of a set of variables and their values, along with its immutable identifier $u.id$.  A node executes a \emph{program} whose actions modify the values of the variables in its local state.  All nodes execute the same program.  Nodes can also communicate with their neighbors.  We use the \emph{synchronous message passing} model of computation~\cite{sync_mp_1998}, where computation proceeds in synchronous rounds.  During each round, a node receives messages sent to it in the previous round from any node in its neighborhood $N(u) = \{v \in V:(u,v) \in E\}$, executes program actions to update its local state, and sends messages to any of its neighbors.  We assume reliable communication channels with bounded delay, meaning a message is received by node $u$ in some round $i$ if and only if it was sent to $u$ in round $i - 1$.

In the overlay network model, nodes communicate over logical links that are part of a node's state, meaning a node may execute actions to create or delete edges in $G$.  In any round, a node may delete any edge incident upon it, as well as create any edge to a node $v$ which has been ``introduced'' to it from some neighbor $w$, such that $(u,w)$ and $(w, v)$ are both in $E$.  Said in another way, in a particular round a node may connect its neighbors to one another by direct logical links.

The goal for our computation is for nodes to execute actions to update their state (including modifying the topology by adding and deleting edges to other nodes) until a legal configuration is reached.  A \emph{legal configuration} can be represented as a predicate over the state of the nodes in the system.  In the overlay network model, links are part of a node's state, meaning a legal configuration is defined at least in part by the overlay network topology.  The \emph{self-stabilizing overlay network problem} is to design an algorithm $\mathcal{A}$ such that when executing $\mathcal{A}$ on each node in a connected network with nodes in an arbitrary state, and allowing $\mathcal{A}$ to add and delete edges, eventually a legal configuration, including a predicate defined at least in part by the network's logical topology, is reached.  This means that a self-stabilizing overlay network will always automatically restore a legal configuration (including reconfiguring the network topology) after \emph{any} transient failure so long as the network remains connected.

\subsection{Complexity Measures}
When designing self-stabilizing overlay network protocols, there are two measures of interest: the time required to build a correct configuration, and the space required to do so (in terms of a node's degree).  In our model, we are concerned with the number of synchronous rounds that are required to reach a legal state.  In particular, the maximum number of synchronous rounds required to build a legal topology when starting from any arbitrary configuration is called the \emph{convergence time}.

When measuring the space requirements, we use the \emph{degree expansion} measure from the original \textsc{Avatar} work~\cite{berns_avatar_15}, which is defined as the ratio of the maximum node degree of any node during convergence over the maximum node degree from the initial or final configuration.  This measure is based upon the idea that if a node begins with a large degree in the initial configuration, or ends with a large degree in the final configuration, the overall algorithm cannot be expected to have a low degree during convergence.  Instead, we are interested in the ``extra'' degree growth caused by the algorithm during convergence.

\section{Generalizing \textsc{Avatar}}
The original \textsc{Avatar} work~\cite{berns_avatar_15} provided two things: a definition of a locally-checkable embedding from any set of real nodes to a particular target topology, and a self-stabilizing algorithm for creating a specific binary tree topology.  Below, we review these contributions and expand the analysis of the algorithm to show it can work for arbitrary topologies.

\subsection{\textsc{Avatar} Definition}
The \textsc{Avatar} network definition is simply a dilation-1 embedding between a \emph{guest network} and a \emph{host network}.  More specifically, let $\mathcal{F}$ be a family of graphs such that, for each $N \in \mathbb{N}$, there is exactly one graph $F_N \in \mathcal{F}$ with node set $\{0, 1, \ldots, N-1\}$.  We call $\mathcal{F}$ a \emph{full graph family}, capturing the notion that the family contains exactly one topology for each ``full'' set of nodes $\{0, 1, \ldots, N-1\}$ (relative to the identifiers).  For any $N \in \mathbb{N}$ and $V \subseteq \{0, 1, \ldots, N-1\}$, $\textsc{Avatar}_{\mathcal{F}}(N,V)$ is a network with node set $V$ that realizes a dilation-1 embedding of $F_N \in \mathcal{F}$.  The specific embedding is given below.  We also show that, when given $N$, \textsc{Avatar} is locally checkable ($N$ can be viewed as an upper bound on the number of nodes in the system).  It is on this full graph family \emph{guest network} that our algorithms shall execute, as we will show later.

\begin{definition}
Let $V \subseteq [N]$ be a node set $\{u_0, u_1, \ldots, u_{n-1}\}$, where $u_i < u_{i+1}$ for $0 \leq i < n-1$.  Let the \emph{range} of a node $u_i$ be $\mathit{range}(u_i) = [u_i, u_{i+1})$ for $0 < i < n-1$.  Let $\mathit{range}(u_0) = [0, u_1)$ and $\mathit{range}(u_{n-1}) = [u_{n-1}, N)$.  $\textsc{Avatar}_{\mathcal{F}}(N, V)$ is a graph with node set $V$ and edge set consisting of two edge types:
\begin{description}
\item[Type 1:] $\{(u_i, u_{i+1}) | i=0,\ldots,n-2\}$
\item[Type 2:] $\{(u_i, u_j) |  u_i \neq u_j \wedge \exists (a,b) \in E(F_N), a \in \mathit{range}(u_i) \wedge b \in \mathit{range}(u_j)\}$
\end{description}
\label{defn:avatar}
\end{definition}

When referring to a general \textsc{Avatar} network for any set of nodes, we will omit the $V$ and simply refer to $\textsc{Avatar}_{\mathcal{F}}(N)$.

As with the original work, to reason about \textsc{Avatar}, we consider two ``networks'': a \emph{host network} consisting of the \emph{real} nodes in $V$, and a \emph{guest network} consisting of the $N$ \emph{virtual} nodes from the target topology.  Each real node in $V$ is the host of one or more virtual nodes in $N$.  This embedding provides several advantages.  First, it allows us to make many networks locally-checkable provided all nodes know $N$ in advance.  Second, it provides a simple mechanism for which to reason about network behavior in the guest network.  As the target $N$-node topology is fixed regardless of the actual set of real nodes $V$, the design and analysis of our algorithms is simplified by executing them on the guest network.  Furthermore, since we are using a dilation-1 embedding, most metrics for performance regarding the guest network (e.g. diameter) still apply to the host network.

As an additional note, we shall assume that the guest nodes also use the synchronous message passing model described for the real nodes in the model section.

\subsubsection{A Note on $n$ versus $N$}
The work of \textsc{Avatar} does require all nodes know $N$, an upper bound on the number of nodes in the system.  In cases where network membership is predictable, it may be possible for $n$ to be within a constant factor of $N$.  From a practical standpoint, even in cases where $N$ and $n$ are significantly different, a polylogarithmic convergence time in $N$ may still be small enough (e.g. if we think of IPv6, $\mathcal{O}(\log N)$ is only 128).

\subsection{The \textsc{Avatar} Algorithm}
The original \textsc{Avatar} work was focused on the creation of a specific topology (a binary search tree) as the \emph{target topology}.  Their algorithm followed a divide-and-conquer approach, separating nodes into clusters and then merging them together.  One can think of their self-stabilizing algorithm as involving three different components:
\begin{enumerate}
    \item \emph{Clustering}: The first step in the algorithm is for nodes to form clusters.  These clusters begin as a single host node hosting a full $N$ node guest network of the target topology (\textsc{Cbt} in the original work).  In the initial configuration, nodes may not be a part of a cluster, but since \textsc{Avatar} is locally checkable, all faulty configurations contain at least one node which detects the faulty configuration and will begin forming the single-node clusters.  This fault detection and cluster creation will propagate through the network until eventually all nodes are members of $N$ node clusters of the target topology.
    \item \emph{Matching}: The second step of the algorithm is to match together clusters so that they may merge together.  To do this, the root node of a spanning tree defined on the cluster repeatedly polls the nodes of its cluster, asking them to either find neighboring clusters that are looking for merge partners (called the leader role), or to look at neighboring clusters that can assign them a merge partner (called the follower role).  The role of leader or follower is randomly selected.  Leader clusters will match together all of their followers for merging by adding edges between the roots of each cluster, creating a matching between clusters that may not be direct neighbors.  This ability to create edges to match non-neighboring clusters allows more matches to occur, and thus more merges, and thus a faster convergence time.
    \item \emph{Merging}: The algorithm then deals with the merging of matched clusters.  To prevent degrees from growing too large, a cluster is only allowed to merge with at most one other cluster at a time.  Once two clusters have matched from the previous step, the roots of the clusters connect as ``partners'' and update their successor pointers based upon the identifier of the host of the root of the other cluster.  One node will have its responsible range become smaller, and this node will send all guest nodes that were in its old responsible range to its partner in the other cluster.  The children of the root nodes are connected, and then they repeat the process of updating successor pointers and passing along guest nodes outside their new responsible range.  Eventually this process reaches the leaves, at which point all nodes in both clusters have updated their responsible ranges and now form a new legal cluster of the target topology.
\end{enumerate} 

As it turns out, the algorithm components from the original \textsc{Avatar} work do not depend upon the specific topology that is being built (the \emph{target topology}).  While the analysis of complexity assumes a complete binary search tree, the algorithm components themselves simply rely upon an arbitrary target topology and a spanning tree defined upon that topology on which to execute PIF waves.  We can therefore extend this algorithm to other topologies if we update the analysis and include several additional metrics.  We define these metrics next after discussing the algorithm's intuition.

\subsection{Relevant Metrics}
To update the analysis of the original \textsc{Avatar} work for any target topology requires two measures of the target topology.  The first of these is diameter of the target topology, which will be a factor in determining the convergence time.  The second of these is a measure of a real node's degree inside the embedded target topology, which will be  a factor in determining the degree expansion.

\subsubsection{Spanning Tree Diameter}
In the original \textsc{Avatar} algorithm, a spanning tree embedded onto the target topology was used to communicate and coordinate between nodes in a particular cluster.  For our work, we will simply use a spanning tree with a root of (virtual) node 0 and consisting of the shortest path from node 0 to all other nodes.  Obviously the diameter of this spanning tree is at most the diameter of the target topology, and we therefore shall use the diameter of the target topology as our first metric of interest.  We denote the diameter of a particular target topology $T$ with $N$ nodes as $D(T_N)$.

As we shall see, this diameter measure will be key in determining the stabilization time.  Intuitively, a low-diameter spanning tree results in faster communication within clusters, and therefore faster convergence than a higher diameter spanning tree.

\subsubsection{Maximum Degree of Embedding}
The other measure of interest has to do with the degree of the \emph{real} nodes when embedding the target topology.  More formally, let the \emph{maximum degree of embedding $T_N$ in \textsc{Avatar}} be defined as the maximum degree of any node in $\textsc{Avatar}_T(N,V)$ for any node set $V \subseteq N$.  Where clear from context, we will refer to this simply as the \emph{maximum degree of embedding} and denote it as $\Delta_\textsc{A}(T_N)$.  Note the maximum degree of embedding is almost entirely determined by the target topology $T$, as there are only 2 edges in $\textsc{Avatar}_T(N)$ per node that are not present to realize a dilation-1 embedding of $T$.

The maximum degree of embedding is a critical measure for degree expansion as it determines how many additional edges a node may receive during the various stages of the algorithm.  The clusters in the \textsc{Avatar} algorithm are $N$ node instances of the target topology $T$, and using metrics defined on $T_N$ is acceptable for running time.  However, each time an edge is added to a virtual node within a cluster, we must consider the effects on the degree in the host network, not just the guest network.

Note that the definition of maximum degree of embedding considers \emph{any} possible subset of real nodes $V$.  This differs from typical (non-stabilizing) overlay network results, where it is common to assume that identifiers are uniformly distributed, meaning that the ranges of each real node are of similar size.  However, since we are building a self-stabilizing protocol, and each cluster is by itself an $N$ node instance of the target topology, the ranges of hosts inside clusters during convergence may be quite skewed, even when the final distribution of node identifiers is not.

\subsection{Overall Complexity}
If we are given the diameter and the maximum degree of embedding of an arbitrary target topology, we can then determine the convergence time and degree expansion of Berns' algorithm for the arbitrary target topology.  We give the theorems for these measures below and provide brief proof sketches of each.  As the updated analysis is heavily based on the original work, and the algorithms are from the original work, the full analysis is left for the appendix.  Our contribution is not the original algorithm of Berns, but rather the observation that the algorithm works for any topology, the updated metrics for the analysis, and the examples and discussion that follow.

\begin{theorem}
\label{thm:avatar_convergence}
The algorithm of Berns~\cite{berns_avatar_15} defines a self-stabilizing overlay network for $\textsc{Avatar}_T$, for some full graph family target topology $T$, with convergence time of $\mathcal{O}(D(T_N) \cdot \log N)$ in expectation, where $D(T_N)$ is the diameter of the $N$ node topology $T_N$.
\end{theorem}
\begin{proofsketch}
A sketch of the steps for proving this theorem are as follows:
\begin{itemize}
    \item In at most $\mathcal{O}(D(T_N))$ rounds, every node is a member of a cluster.
    \item For any cluster, in an expected $\mathcal{O}(D(T_N))$ rounds, the cluster has completed a merge with another cluster, meaning the number of clusters has decreased by a constant fraction in $\mathcal{O}(D(T_N))$ rounds.
    \item Reducing the number of clusters by a constant fraction needs to be done $\mathcal{O}(\log N)$ times before a single cluster remains.
\end{itemize}
\end{proofsketch}

\begin{theorem}
\label{thm:avatar_degree-expansion}
The algorithm of Berns~\cite{berns_avatar_15} defines a self-stabilizing overlay network for $\textsc{Avatar}_T$, for some full graph family target topology $T$, with degree expansion of $\mathcal{O}(\Delta_\textsc{A}(T_N) \cdot \log N)$ in expectation, where $\Delta_\textsc{A}(T_N)$ is the maximum degree of embedding of the $N$ node target topology $T_N$.
\end{theorem}
\begin{proofsketch}
To prove this, we consider the actions that might increase a node's degree.
\begin{itemize}
    \item Regardless of the number of merges a node participates in, the node's degree will grow to at most $\mathcal{O}(\Delta_\textsc{A}(T_N))$ as the result of merge actions.  By definition, a node's degree within its cluster after a merge cannot exceed $\Delta_\textsc{A}(T_N)$.
    \item During the process of matching clusters together, a node's degree may grow by one for every child it has in the spanning tree.  Since the node has at most $\Delta_\textsc{A}(T_N)$ children from other nodes, each time the node participates in the matching its degree grows by $\mathcal{O}(\Delta_\textsc{A}(T_N))$.  In expectation, this matching happens $\mathcal{O}(\log N)$ times (see proof of Theorem \ref{thm:avatar_convergence}).
\end{itemize}
\end{proofsketch}

The implications of Theorems~\ref{thm:avatar_convergence} and \ref{thm:avatar_degree-expansion} are that we can simply define a target topology and analyze its diameter and maximum degree of embedding to have a self-stabilizing protocol for our target topology.  We provide a few examples of this process in the following section.
\section{Examples}
In this section, we demonstrate how the selection of the target topology affects the complexity of our algorithm by considering several different topologies: the \textsc{Linear} network, a complete binary search tree (\textsc{Cbt}, taken from~\cite{berns_avatar_15}), and \textsc{Chord}~\cite{stoica_chord_01}.

\subsection{Linear}
As the name suggests, the \textsc{Linear} network consists of a line of nodes sorted by identifier.  The formal desired end topology for the \textsc{Linear} network is given next.

\begin{definition}
The $\textsc{Linear}(N)$ network, for $N \in \mathbb{N}$, consists of nodes $V = \{0, 1, \ldots, N-1\}$ and edges $E = \{(i, i+1), i \in [0,N-2]\}$.
\end{definition}

\begin{lemma}
The diameter of an $N$ node \textsc{Linear} network is $\mathcal{O}(N)$.
\end{lemma}

Each virtual node has a degree of at most 2, and it is easy to see that each real node also has a degree of at most 2.  Therefore, the maximum degree of embedding for \textsc{Linear} is $\mathcal{O}(1)$.

\begin{lemma}
  The maximum degree of embedding of the \textsc{Linear} topology is $\mathcal{O}(1)$.
\end{lemma}
\begin{proof} 
Note that for any particular $range(u) = [x, y]$, there are at most two external edges: $(x-1, x)$ and $(y, y+1)$.  All other edges are between virtual nodes inside the range.
\end{proof}

The above lemmas combined with Theorems \ref{thm:avatar_convergence} and \ref{thm:avatar_degree-expansion} give us the following corollary.

\begin{corollary}
The self-stabilizing \textsc{Avatar} algorithm builds $\textsc{Avatar}_\textsc{Linear}(N)$ in an expected $\mathcal{O}(N \cdot \log N)$ rounds with an expected degree expansion of $\mathcal{O}(\log N)$.
\end{corollary}

Note the convergence time of this algorithm is a logarithmic factor slower than previous results~\cite{onus_linear_07}.  This logarithmic factor comes from the ``cost of coordination'', as edges are only added when clusters have matched.

\subsection{Complete Binary Search Tree}
In the first \textsc{Avatar} paper, the author defined and analyzed an algorithm for one specific topology, the complete binary search tree (called \textsc{Cbt}).  We formally define the desired end topology for \textsc{Cbt}, list the relevant measures for this topology below and omit the proofs, as those are contained in the work of Berns~\cite{berns_avatar_15}.

\begin{definition}
For $a \leq b$, let $\textsc{Cbt}[a,b]$ be a binary tree rooted at $\mathit{r} = \lfloor (b+a)/2 \rfloor$.  Node $r$'s left cluster is $\textsc{Cbt}[a,r-1]$, and $r$'s right cluster is $\textsc{Cbt}[r+1,b]$.  If $a > b$, then $\textsc{Cbt}[a,b] = \bot$.  We define $\textsc{Cbt}(N) = \textsc{Cbt}[0,N-1]$.
\end{definition}

\begin{lemma}
  The diameter of an $N$ node \textsc{Cbt} network is $\mathcal{O}(\log N)$.
\end{lemma}

\begin{lemma}
  The maximum degree of embedding of an $N$ node \textsc{Cbt} network is $\mathcal{O}(\log N)$.
\end{lemma}

The above lemmas combined with Theorems \ref{thm:avatar_convergence} and \ref{thm:avatar_degree-expansion} give us the following corollary.

\begin{corollary}
The self-stabilizing \textsc{Avatar} algorithm builds a target topology of $\textsc{Avatar}_\textsc{Cbt}(N)$ in expected $\mathcal{O}(\log^2 N)$ rounds with $\mathcal{O}(\log^2 N)$ expected degree expansion.
\end{corollary}

Note the above corollary matches with the detailed proofs given in the original \textsc{Avatar} work.  Unlike the original, however, we reached our conclusions based simply upon the metrics we defined earlier.  This corollary, then, serves as a nice ``sanity check'' on the accuracy of our results.

\subsection{Chord}
Both \textsc{Linear} and \textsc{Cbt} are tree topologies, meaning they are fragile in the sense that a single node or link failure may partition the network.  In this section, we consider a more robust topology.  In particular, we apply the \textsc{Avatar} algorithm to an $N$-node \textsc{Chord} network~\cite{stoica_chord_01} as well.  We define the network's desired topology as follows.

\begin{definition}
For any $N \in \mathbb{N}$, let $\textsc{Chord}(N)$ be a graph with nodes $[N]$ and edge set defined as follows.  For every node $i$, $0 \leq i < N$, add to the edge set $(i, j)$, where $j = (i + 2^k) \mod N$, $0 \leq k < \log N - 1$.  When $j = (i + 2^k)\mod N$, we say that $j$ is the $k$-th finger of $i$.
\end{definition}

The original \textsc{Chord} paper proves the following lemma regarding the network's diameter.

\begin{lemma}
The diameter of an $N$ node \textsc{Chord} network is $\mathcal{O}(\log N)$.
\end{lemma}

While the logarithmic diameter means we can efficiently build \textsc{Chord} in terms of time complexity, the results are not so hopeful in terms of degree complexity.  In an $N$-node \textsc{Chord} network, every node has $\mathcal{O}(\log N)$ neighbors, some of which have identifiers up to $N/2$ away.  The result of this is that if a real node has a range of size $N/2$, each virtual node in the range may potentially have a connection to a virtual node on a different host.  This means that a real node in an embedded \textsc{Chord} network may have a maximum degree of embedding of $\mathcal{O}(N)$, as we show next.

\begin{lemma}
  The maximum degree of embedding of the \textsc{Chord} topology is $\mathcal{O}(N)$.
\end{lemma}
\begin{proof} 
To see this, consider a specific $n$ node embedding of \textsc{Chord} where $range(u) = [0, N/2)$.  Note each node in $u$'s range is incident on at least one edge whose other endpoint is outside $range(u)$ -- specifically, the $\log N - 1$ \textsc{Chord} finger.  As there are $N/2$ nodes in $u$'s range, we have at least $N/2$ edges with exactly one endpoint outside $range(u)$, and our lemma holds.
\end{proof}

Note this result is not a concern in the \emph{final} configuration if we assume identifiers are uniformly distributed (which is a reasonable and common assumption with overlay networks).  However, we are working in a self-stabilizing setting where \emph{any} initial configuration is possible, and therefore one could imagine a scenario where the system reaches a configuration where a single cluster consisting of node $0$ is matched with an $N/2$ node cluster consisting of all real nodes from the range $[N/2, N)$.

Combining the above two lemmas with Theorem \ref{thm:avatar_convergence} and Theorem \ref{thm:avatar_degree-expansion} gives us the following corollary.

\begin{corollary}
\label{corollary:chord}
The self-stabilizing \textsc{Avatar} algorithm builds $\textsc{Avatar}_{\textsc{Chord}}(N)$ in expected $\mathcal{O}(\log^2 N)$ rounds and with expected degree expansion of $\mathcal{O}(N \cdot \log N)$.
\end{corollary}

Note we can actually improve the bound of the degree expansion, as it is at most $\mathcal{O}(N)$.  As our theorems provide an upper bound, however, we leave them stated as is for simplicity.
\section{\textsc{SkipChord}}
\label{section:skipchord}
We showed above the role the diameter and maximum degree of embedding play in determining performance of the \textsc{Avatar} algorithm.  One of the benefits of the general analysis of the \textsc{Avatar} algorithm is that it highlights the factors of the \emph{topology} that will affect convergence, allowing a network designer to select a topology based upon the problem requirements while weighing the impact of the topology on convergence time and degree expansion.  One can even design \emph{new} topologies with an eye towards these metrics for use in \textsc{Avatar} embeddings.

In this section, we present a new network topology built specifically for embedding in the \textsc{Avatar} framework which we call \textsc{SkipChord}.  This ring-based network is more robust than \textsc{Linear} and \textsc{Cbt} while avoiding the high degree requirements during stabilization of the standard \textsc{Chord} network.  As we shall show, by balancing degree and robustness, we may achieve efficient stabilization while still having sublinear degree expansion.

\subsection{Definition}
As trees, the \textsc{Linear} and \textsc{Cbt} networks make poor choices for many fault-prone applications as they are easily disconnected by node or link failure.  \textsc{Chord} represents a more robust choice, but suffers from a high maximum degree of embedding due to each of the $N$ nodes having a long link (to a neighbor with identifier $\mathcal{O}(N)$ away from itself).  Our \textsc{SkipChord} network tries to balance this by limiting the number of fingers while still maintaining a topology more robust than a simple tree.  We give the formal definition of \textsc{SkipChord} below.

\begin{definition}
For any $N \in \mathbb{N}$, let $\textsc{SkipChord}(N,s)$ be a graph with nodes $[N]$ and \emph{skip factor} $s$.  The edge set for $\textsc{SkipChord}(N,s)$ is defined as follows:
\begin{itemize}
    \item \emph{Ring edges}: For every node $i$, $0 \leq i < N$, add edges $(i-1 \mod N, i)$ and $(i, i+1 \mod N)$
    \item \emph{Finger edges}: For every node $j$, where $j = sk$, for $k=0,1,\ldots,(N-s)/s$, add edge $(j, j + 2^{k \mod \log N} \mod N)$.  We say the \emph{size} of this finger edge is $2^{k \mod \log N}$.
\end{itemize}
\end{definition}

This construction basically takes the $\log N$ fingers from each node in the original \textsc{Chord} and distributes them out over a range of nodes as determined by the \emph{skip factor} $s$.  As we shall show, by ``skipping'' the fingers, fewer virtual nodes in a real node's range have ``long'' outgoing links, and therefore the number of edges to other real nodes is limited while not compromising efficient routing.

\begin{figure}
\centering
\subfigure[Six Fingers of Node $0$ in $\textsc{Chord}(64)$]{
   \includegraphics[scale=0.25]{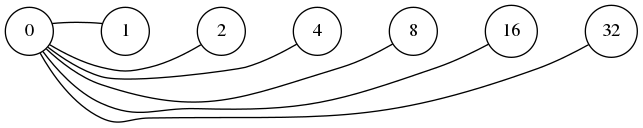}
}\hfill
\subfigure[First Six ``Skipped'' Fingers for $\textsc{SkipChord}(64,2)$]{
   \includegraphics[scale=0.25]{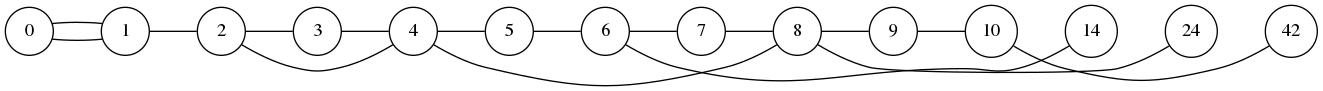}
}
\caption{(a) The neighborhood of node $0$ in \textsc{Chord}, and (b) the corresponding six fingers in \textsc{SkipChord}.  Note how the fingers in \textsc{SkipChord} are no longer all incident on node $0$ but instead have ``skipped'' ahead.}
\label{fig:skipchord}
\end{figure}

To better understand \textsc{SkipChord}, consider Figure~\ref{fig:skipchord}.  The network on the top shows node $0$'s neighborhood for $\textsc{Chord}(64)$, while the network on the bottom shows the corresponding fingers ``skipped'' with a skip factor of 2 (i.e. a subset of the edges for $\textsc{SkipChord}(64,2)$).  The first six fingers in \textsc{Chord} are all incident upon node $0$, while the first six fingers in \textsc{SkipChord} are distributed amongst nodes 0, 2, 4, 6, 8, and 10.  As a result of this, each node in \textsc{SkipChord} has a much smaller degree than in \textsc{Chord}.  This low degree result holds true for the degree of the embedding as well, as we show in the following section.

\subsection{Metrics}
We begin with a proof of the diameter of \textsc{SkipChord}.  The intuition behind our result is simple: any node is at most $s \cdot \log N$ away from an edge that at least halves the distance from itself to any other node.

\begin{lemma}
The diameter of an $N$-node \textsc{SkipChord} network with skip factor $s$ (\textsc{SkipChord}$(N,s)$) is $\mathcal{O}(s \cdot \log^2 N)$.
\end{lemma}
\begin{proof}
Consider the hops required to halve the distance between two arbitrary nodes $u$ and $v$.  In at most $\mathcal{O}(s \cdot \log N)$ hops from $u$ using ring edges, at least one finger of every size is reachable.  One of these fingers will at least halve the distance to the other node $v$.  As the distance to $v$ can be halved in $\mathcal{O}(s \cdot \log N)$ hops, and this halving will occur $\log N$ times before reaching $v$, our lemma holds.
\end{proof}

Next we consider the maximum degree of embedding.  Here we see how spreading the fingers out over a set of nodes has resulted in a lower maximum degree of embedding.

\begin{lemma}
The maximum degree of embedding of an $N$-node \textsc{SkipChord} with a skip factor of $s$ (\textsc{SkipChord}$(N,s)$) is $\mathcal{O}(\frac{N}{s \cdot \log N})$.
\end{lemma}
\begin{proof}
Let an edge be called an \emph{external edge} for node $u$ if and only if it has exactly one endpoint in $range(u)$.  Note that there are at most 2 external ring edges for any possible range for node $u$.  We consider then the external finger edges in $range(u)$.  Let $k$ be the largest finger that has exactly one endpoint in $range(u)$.  There are at most $|range(u)| / s \cdot \log N$ such fingers, where $|range(u)|$ denotes the size of the range (for our embeddings, the number of nodes $u$ is hosting).  For the $k-1$ fingers, there are at most half as many as the $k$ fingers with exactly one endpoint in $range(u)$.  Similarly, for the $k-2$ fingers, there are at most $1/4$ as many as the $k$ fingers with exactly one endpoint in $range(u)$, and so on.  Summing these together, we get $(|range(u)| / (s \cdot \log N)) + 1/2(|range(u)| / (s \cdot \log N)) + 1/4(|range(u)| / (s \cdot \log N)) + \ldots = 2(|range(u)| / (s \cdot \log N))$.  Since $|range(u)|$ is at most $N$, our lemma holds.
\end{proof}

The above lemmas, combined with Theorem~\ref{thm:avatar_convergence} and Theorem~\ref{thm:avatar_degree-expansion} give us the following corollary.

\begin{corollary}
\label{corollary:skipchord}
The \textsc{Avatar} algorithm builds the $\textsc{SkipChord}(N,s)$ target topology in an expected $\mathcal{O}(s \cdot \log^3 N)$ rounds with an expected degree expansion of $\mathcal{O}(N/s)$.
\end{corollary}

Note that we can select a skip factor in such a way as to have efficient time and space complexity.  For instance, if we select a skip factor of $\log N$, we have polylogarithmic convergence time, sublinear degree expansion, and a target topology that is more robust than the tree topologies of \textsc{Linear} and \textsc{Cbt}.

Given the fact that \textsc{SkipChord} can be built efficiently, there are multiple suitable applications for it.  Like \textsc{Chord}, \textsc{SkipChord} can be used as a distributed hash table for storing and retrieving files, particularly in settings where transient failures may cause node and link failures, such as with unreliable Internet connections or in scenarios where a large number of nodes may be deployed and start up at the same time with the goal of forming a distributed hash table.  If node identifiers are uniformly distributed, each real node would host a similar number of files as with the original \textsc{Chord} while having the added benefit of efficient stabilization after any transient faults.

\section{Discussion and Future Work}
Table~\ref{table:measures} summarizes the various measures of interest for the different target topologies.

\begin{table}
\begin{center}
\begin{tabular}{|c|c|c|}
  \hline
  \textbf{Topology} & \textbf{Diameter} & \textbf{Maximum Degree of Embedding} \\
  \hline
  \textsc{Linear} & $\mathcal{O}(N)$ & $\mathcal{O}(1)$ \\
  \hline
  \textsc{Cbt} & $\mathcal{O}(\log N)$ & $\mathcal{O}(\log N)$ \\
  \hline
  \textsc{Chord} & $\mathcal{O}(\log N)$ & $\mathcal{O}(N)$ \\
  \hline
  \textsc{SkipChord} & $\mathcal{O}(s \cdot \log^2 N)$ & $\mathcal{O}(N / (s\cdot \log N))$ \\
  \hline
\end{tabular}
\end{center}
\caption{Relevant Metrics for Various Overlay Topologies}
\label{table:measures}
\end{table}

Besides providing a simple way to build and analyze self-stabilizing overlay networks, our analysis provides a set of parameters that a designer can tune to achieve a target level of efficiency.  One application of our work, then, is in guiding the creation of new topologies that use the \textsc{Avatar} embedding and strive for low diameter \emph{and} low maximum degree of embedding while still maintaining other desirable properties like robustness to node or link failure.  We have demonstrated this process with the creation of the \textsc{SkipChord} topology.  It would be interesting to see how other topologies perform in this framework.

Future work could also consider variations to our framework.  For instance, currently we have only considered spanning trees created by finding the shortest path from $0$ to every node.  Note, however, that this would not have to be the case.  One could imagine spanning trees of larger diameter but smaller maximum degree of embeddings that might help lower the degree complexity of the algorithm.

Finally, our framework can be used to better understand the upper and lower bounds for the work or degree expansion of self-stabilizing overlay network protocols.  While our results deal entirely with network embeddings, it would be interesting to see if the provided insights help make general bounds for any network, embedded or not.

\bibliographystyle{splncs04}
\bibliography{overlays}

\appendix
\appendixpage
\let\cleardoublepage\clearpage
\section{Updated Analysis of Berns~\cite{berns_avatar_15}: Convergence Time}
In this appendix, we restate the original algorithms of Berns~\cite{berns_avatar_15} and provide an updated analysis for any topology.  As stated in the original paper, our analysis depends upon two measures: the diameter of a shortest-path spanning tree rooted at virtual node $0$ embedded onto the target topology $T$ ($D(T_N)$), and the maximum degree of embedding of the target topology $T$ ($\Delta_\textsc{A}(T_N)$).

\subsection{Clustering}
The first step in the algorithm is for real nodes to form clusters.  These clusters will be used to coordinate convergence actions so that convergence time and degree expansion may remain limited.  We shall use as our cluster topology an embedding of the target topology onto a subset of nodes.  More specifically, we shall say a group of nodes $V'$ is part of a cluster if the subgraph induced by $V'$ is a legal configuration of the target network.

To manage cluster membership, we introduce into each node's state a \emph{cluster identifier}, which we shall set as the identifier of the real node which hosts the root of the spanning topology for that cluster.

Clusters are used to coordinate the convergence process to limit degree expansion while also ensuring fast convergence.  This intra-cluster coordination requires a mechanism to communicate with all nodes in a cluster.  As was done with $\textsc{Avatar}_\textsc{Cbt}$, we will use a non-snap-stabilizing variant of the propagation of information with feedback and cleaning (PFC) algorithm~\ref{algo:pfc}.  Communication happens through PFC waves that are initiated by the cluster's root.  The information is sent level-by-level until reaching the leaves, at which point a feedback wave returns back up the tree to the root, collecting any information regarding a reply and ``cleaning'' the tree to prepare it for the next wave of communication.

\subsubsection{Algorithms}
Algorithm \ref{algo:reset} presents the simple reset algorithm nodes execute to ensure membership in a cluster.

\begin{algorithm}
\begin{tabbing}
........\=....\=....\=....\=....\=....\=....\kill
1.\>\textbf{if} a \emph{reset fault} is detected \textbf{and}\\
\>\>did not reset in previous round \textbf{then}\\
2.\>\>$\mathit{clusterSucc}_u = \bot$; $\mathit{clusterPred}_u = \bot$; $\mathit{cluster}_u = \mathit{id}_u$;\\
3.\>\textbf{fi}
\end{tabbing}
\caption{The Reset Algorithm}
\label{algo:reset}
\end{algorithm}

The modified $PFC$ algorithm which will be executed on the spanning tree of the guest network is given in Algorithm \ref{algo:pfc}.  We use $I$ to represent the information being passed down, and $F$ to represent the information that should be fed back.  We also introduce the idea of a \emph{feedback action}, which allows us to specify node behavior that should be executed before completing the feedback portion of the $PFC$ wave.  In particular, we will use this action for a leader cluster to ``wait'' for followers to be ready to receive a merge partner, as well as for maintaining the faulty bit after a merge.

\begin{algorithm}
\begin{tabbing}
........\=....\=....\=....\=....\=....\=....\kill
\textbf{Variable for Node $a$:} $\mathit{PFCState}_a$\\
1.\>\textbf{when} node $\mathit{root}$ satisfies $\mathit{PFCState}_{\mathit{Children}(\mathit{root})} = \mathit{PFCState}_{\mathit{root}} = \mathit{Clean}$ \textbf{then}\\
2.\>\>$\mathit{root}_T$ initiates $\mathit{PFC}$ wave by setting $\mathit{PFCState}_{\mathit{root}} = \mathit{Propagate}(I)$\\
3.\>\>Each node $a$ executes the following:\\
4.\>\>\>\textbf{if} $\mathit{PFCState}_a = \mathit{Clean} \wedge \mathit{PFCState}_{\mathit{Parent}(a)} = \mathit{Propagate}(I) \wedge$\\
\>\>\>\>$\mathit{PFCState}_{\mathit{Children}(a)} = \mathit{Clean}$ \textbf{then}\\
5.\>\>\>\>\>$\mathit{PFCState}_a = \mathit{Propagate}(I)$\\
6.\>\>\>\textbf{else if} $\mathit{PFCState}_a = \mathit{Propagate}(I) \wedge \mathit{PFCState}_{\mathit{Parent}(a)} = \mathit{Propagate}(I) \wedge$\\
\>\>\>\>$\mathit{PFCState}_{\mathit{Children}(a)} = \mathit{Feedback}(F)$ \textbf{then}\\
7.\>\>\>\>\>$\mathit{PFCState}_a = \mathit{Feedback}(F')$\\
8.\>\>\>\textbf{else if} $\mathit{PFCState}_a = \mathit{Feedback}(F) \wedge \mathit{PFCState}_{\mathit{Parent}(a)} = \mathit{Feedback}(F) \wedge$\\
\>\>\>\>$\mathit{PFCState}_{\mathit{Children}(a)} = \mathit{Clean}$ \textbf{then}\\
9.\>\>\>\>\>$\mathit{PFCState}_a = \mathit{Clean}$\\
10.\>\>\>\textbf{fi}\\
11.\>\textbf{fi}
\end{tabbing}
\caption{Subroutine for $\mathit{PFC}(I,F)$ for spanning tree on guest network}
\label{algo:pfc}
\end{algorithm}

We can use the full definition of the $PFC$ mechanism to define the following particular type of cluster.

\begin{definition}
A set of nodes $C$ is called a \emph{proper cluster} for target topology $T$ when, for each $u \in C$, (i) the neighborhood of $u$ induced by nodes in $C$ matches $u$'s neighborhood in $\textsc{Avatar}_T(N,C)$, (ii) cluster successors and predecessors of $u$ are consistent with $u$'s successor and predecessor in the graph induced by nodes in $T$, (iii) all nodes in $C$ have the same correct cluster identifier, (iv) no node in $C$ neighbors a node $v \notin C$ such that $\mathit{cluster}_v = \mathit{cluster}_u$, and (v) the communication mechanism is not faulty.
\end{definition}

\subsubsection{Analysis of Clustering}
Given the definitions and algorithms above, we provide the full analysis of the clustering portion of our algorithm.

\begin{lemma}
If the root of a proper cluster $C$ initiates the $\mathit{PFC}(I, F)$ wave, the $\mathit{PFC}$ wave is complete (all nodes receive the information, the root receives the feedback, and nodes are ready for another $\mathit{PFC}$ wave) in $\mathcal{O}(D(T_N))$ rounds, where $D(T_N)$ is the diameter of the $N$-node target topology $T$.
\label{dlemma:pfc}
\end{lemma}
\begin{proof}
After the root initiates the $\mathit{PFC}(I,F)$ wave, in every round the information $I$ moves from level $k$ to $k+1$ until reaching a leaf.  As the cluster has $\mathcal{O}(D(T_N))$ levels, after $\mathcal{O}(D(T_N))$ rounds, all nodes have received the propagation wave and associated information.  Upon receiving the propagation wave, leaves set their $\mathit{PFC}$ states to $\mathit{Feedback}$ to begin the feedback wave.  Again, in every round the feedback wave moves one level closer to the root, yielding an additional $\mathcal{O}(D(T_N))$ rounds before the root has received the feedback wave.  Consider the transition to $\mathit{PFC}$ state $\mathit{Clean}$.  When a leaf node $b$ sets its $\mathit{PFC}$ state to $\mathit{Feedback}$, in one round the parent of $b$ will set its state to $\mathit{Feedback}$, and in the second round, $b$ will set its state to $\mathit{Clean}$.  The process then repeats for the parent of $b$.  In general, two rounds after a node $b$ is in state $\mathit{Feedback}$, it transitions to state $\mathit{Clean}$.  Therefore, $\mathcal{O}(D(T_N))$ rounds after the root initiates a $\mathit{PFC}$ wave, all nodes receive the propagation wave, return the feedback wave, and move back to a clean state, ready for another $\mathit{PFC}$ wave.
\end{proof}

\begin{lemma}
Let node $b$ be a member of a proper cluster $C$.  Node $b$ can only execute a reset action if $C$ begins the merging process from Algorithm \ref{algo:merge}.
\label{dlemma:merge_reset}
\end{lemma}
\begin{proof}
To begin, notice that the cluster structure is not changed unless a node executes a merge action.  Furthermore, the $\mathit{PFC}$ algorithm will never cause a reset fault in a proper cluster.  Therefore, no reset is executed unless a merge action is executed.
\end{proof}

\begin{lemma}
Let a virtual node $r_C$ be a node whose immediate neighborhood matches the neighborhood of a correct root node ($r_C$ has a consistent $\mathit{PFC}$ state and correct cluster neighborhood).  If $r_C$ initiates a $\mathit{PFC}(I, F)$ wave and later receives the corresponding feedback wave, then $\mathit{r}_C$ is the root of a proper cluster.
\label{dlemma:pfc_works}
\end{lemma}
\begin{proof}
Every node will only continue to forward the propagate and feedback waves if (i) they have the appropriate cluster neighbors, and (ii) their $\mathit{PFC}$ states are consistent.  Therefore, if $r_C$ receives a feedback wave, $C$ must be a proper cluster.
\end{proof}

We classify clusters based upon their $\mathit{PFC}$ state next.

\begin{definition}
Let $C$ be a proper cluster.  $C$ is a \emph{proper clean cluster} if and only if all nodes in $C$ have a $\mathit{PFC}$ state of $\mathit{Clean}$.
\end{definition}

Let $\mathcal{F}(G_i)$ be all configurations reached by executing our algorithm beginning in configuration $G_i$ (that is, the set of future configurations).  We now bound the  occurrences of reset after a given configuration.

\begin{lemma}
If node $b$ is a member of a proper unmatched clean cluster $C$ in configuration $G_i$, then $b$ will never execute a reset in any configuration $G_j \in \mathcal{F}(G_i)$.
\label{dlemma:clean_merge}
\end{lemma}
\begin{proof}
By Lemma \ref{dlemma:merge_reset}, only a merge can cause node $b$ to execute a reset fault.  We show that any merge $b$ participates in must be between two proper clusters, and therefore completes correctly (Lemma \ref{dlemma:merge}).

Suppose the root of $C$ has been matched with the root of another cluster $C'$.  $C$ cannot begin modifying its cluster edges for the merge until both $C$ and $C'$ have successfully completed the $\mathit{PFC}(Prep(C,C'),\bot)$ wave.  Suppose $C'$ was not a proper cluster.  In this case, $C'$ would not successfully complete the $\mathit{PFC}$ wave (Lemma \ref{dlemma:pfc_works}), and $C$ would not begin a merge with $C'$.  Therefore, if $C$ and $C'$ merge together, both must be proper clusters, implying the merge completes successfully and $b$ is again a member of a clean proper cluster $C''$ (see Lemma \ref{dlemma:merge}).
\end{proof}

\begin{lemma}
Consider a node $b$ that is not a member of a proper cluster in configuration $G_i$.  In $\mathcal{O}(D(T_N))$ rounds, $b$ is a member of a proper unmatched clean cluster.
\label{dlemma:reset_to_clean}
\end{lemma}
\begin{proof}
If $b$ is not a member of a proper cluster and detects a reset fault in $G_i$, our lemma holds.

Consider the case where $b$ is not a member of a proper cluster but has no reset fault.  If $b$'s cluster neighborhood is not a legal cluster neighborhood, then $b$ must be performing a merge between its own cluster $C$ and neighboring cluster $C'$.  In 2 rounds, $b$ either has the correct cluster neighbors, each with tree identifier of either $C$ or $C'$, and has passed the merge process on to its children, or $b$ has executed a reset.  The children of $b$ now either execute a reset, or are in a state consistent with the merge process, and we repeat the argument.  As there are $\mathcal{O}(D(T_N))$ levels, after $\mathcal{O}(D(T_N))$ rounds either a node has reset due to this merging, or the merge is complete.  If a node has reset in round $i$, its parent will reset in round $i+1$, its parent's parent will reset in $i+2$, and so on.  After at most $\mathcal{O}(D(T_N))$ rounds, $b$ resets and becomes part of an unmatched clean proper cluster.

If $b$ does not detect locally that it is not a member of proper cluster $C$, then there must exist a node $c$ within distance $\mathcal{O}(D(T_N))$ such that every node $p_i$ on a path from $b$ to $c$ believes it is part of the same cluster as $b$, and node $c$ detects an incorrect cluster neighborhood or inconsistent $\mathit{PFC}$ state.  As with the above argument, either $c$ detects a reset fault immediately, or $c$ is participating in a merge.  In either case, after $\mathcal{O}(D(T_N))$ rounds, $c$ has either reset, or $c$ is a member of a proper unmatched clean cluster resulting from a successful merge.
\end{proof}

Combining Lemmas \ref{dlemma:clean_merge} and \ref{dlemma:reset_to_clean} gives us the following lemma.

\begin{lemma}
\label{lemma:no_more_resets}
No node executes a reset action after $\mathcal{O}(D(T_N))$ rounds.
\end{lemma}

We call a configuration $G_i$ a \emph{reset-free configuration} if and only if no reset actions are executed in any configuration $G_j \in \mathcal{F}(G_i)$.  For the remainder of our proofs in this appendix, we shall assume a reset-free configuration.

\subsection{Matching}
In \textsc{Avatar}, convergence to the target topology happens through the merging of clusters.  To limit degree expansion, a cluster may be merging with at most one other cluster at any particular point in time.  There are cluster topologies, however, where a maximum matching is of size $\mathcal{O}(1)$ (for instance, the star topology), which would therefore slow convergence.  We use the matching algorithm from \textsc{Avatar} to match together clusters within distance 2 of one another to thus ensure efficient convergence.  The general idea of this matching algorithm is simple: clusters randomly select a role (either leader or follower) and then leaders pair their followers to allow them to merge.

\subsubsection{Algorithms}
Below we include the full algorithm for the matching process done by a leader cluster (Figure \ref{algo:lead}), the subroutine used by leaders to create the matching among followers (Figure \ref{algo:connect}), and the algorithm followed by a follower cluster (Figure \ref{dalgo:follow}).

\begin{algorithm}
\begin{tabbing}
........\=....\=....\=....\=....\=....\=....\kill
\>\textit{// The root $r_C$ of $T$ has selected the leader role}\\
1.\>Node $r_C$ uses $PFC$ to inform all nodes in $C$ of\\
\>\>leader role\\
2.\>Node $r_C$ uses $PFC$ to close all nodes in $C$\\
3.\>Node $r_C$ initiates the $\mathit{ConnectFollowers}$ procedure (Algorithm \ref{algo:connect})\\
4.\>\textbf{if} $C$ was not assigned a merge partner \\
\>\>\>during $\mathit{ConnectFollowers}$ \textbf{then}\\
5.\>\>$r_C$ randomly selects a new role\\
6.\>\textbf{fi}
\end{tabbing}
\caption{The Matching Algorithm for a Leader Cluster $C$}
\label{algo:lead}
\end{algorithm}

\begin{algorithm}
\begin{tabbing}
........\=....\=....\=....\=....\=....\=....\kill
1.\>Execute $\mathit{PFC}(\mathit{ConnectFollowers}, \bot)$:\\
2.\>\>\textbf{Feedback Action for $a$:}\\
3.\>\>\>\textbf{while} $\exists b \in N_a : \mathit{role}_b = \mathit{PotentialFollower}(a) \vee $\\
\>\>\>\>$(\mathit{role}_b = \mathit{Follower}(a) \wedge b \neq \mathit{root})$ \textbf{do}\\
4.\>\>\>\>skip;\\
5.\>\>\>\textbf{od}\\
4.\>\>\>Order the $k$ followers in $N_a$ by tree\\
\>\>\>\> identifiers $b_0, b_1, b_2, \ldots, b_{k-1}$\\
5.\>\>\>\textbf{for} $i = 0,2,4,\ldots,\lfloor k/2 \rfloor$ \textbf{do}\\
6.\>\>\>\>Create edge $(b_i, b_{i+1})$;\\
7.\>\>\>\>Set merge partner of $b_i$ to $b_{i+1}$ and vice versa\\
7.\>\>\>\>Delete edge $(a,b_{i+1})$\\
8.\>\>\>\textbf{od}\\
9.\>\>\>\textbf{if} $k \bmod 2 \neq 0$ \textbf{then}\\
10.\>\>\>\>Create edge $(\mathit{parent}_a,b_{k-1})$\\
11.\>\>\>\>Delete edge $(a,b_{k-1})$\\
11.\>\>\>\textbf{fi}
\end{tabbing}
\caption{Subroutine $\mathit{ConnectFollowers}$}
\label{algo:connect}
\end{algorithm}

\begin{algorithm}
\begin{tabbing}
........\=....\=....\=....\=....\=....\=....\kill
\>\textit{// Assume root $r_C$ of $C$ has selected follower role}\\
1.\>Root $r_C$ selects a role of \emph{short} or \emph{long} follower uniformly at random.\\
2.\>Root $r_C$ propagates role of follower to nodes in $C$.\\
3.\>\textbf{if} $C$ is a \emph{short follower} \textbf{then}\\
4.\>\>Root $r_C$ sets $\mathit{pollCnt} = 2$\\
5.\>\textbf{else} $C$ is a \emph{long follower}:\\
6.\>\>Root $r_C$ sets $\mathit{pollCount} = 12$\\
7.\>\textbf{fi}\\
8.\>\textbf{while} $\mathit{pollCnt} > 0$ \textbf{and} no potential leader found \textbf{do}\\
9.\>\>Root $r_C$ queries cluster $C$ (using $PFC$) for a potential leader.\\
10.\>\>\textbf{if} $u \in C$ found a potential leader $v \in C'$ \textbf{then}\\
11.\>\>\>Forward one edge to $v$ to parent of $u$ during feedback wave\\
12.\>\>\textbf{fi}\\
13.\>\>Root $r_C$ sets $\mathit{pollCnt} = \mathit{pollCnt} - 1$;\\
14.\>\textbf{od}\\
15.\>\textbf{if} a potential leader is returned to $r_C$ \textbf{then}\\
16.\>\>Root node $r_C$ selects one potential leader $v \in C'$\\
\>\>\>and informs nodes in $C$ of $v$\\
17.\>\textbf{else} \\
18.\>\>$r_C$ randomly selects a new role from $(leader, follower)$\\
19.\>\textbf{fi}
\end{tabbing}
\caption{The Matching Algorithm for a Follower Cluster $C$}
\label{dalgo:follow}
\end{algorithm}

\subsubsection{Analysis of Matching}
\begin{lemma}
Let $b \in C$ be a follower that has selected a neighbor $c \in C'$ as a potential leader.  In $\mathcal{O}(D(T_N))$ rounds, the root of $C$ has an edge to some leader cluster $C''$, and all nodes in $C$ know this leader.
\label{dlemma:follow_to_root}
\end{lemma}
\begin{proof}
When $b$ detects $c$ becomes a potential leader, $b$ marks $c$ as a potential leader immediately, regardless of the $\mathit{PFC}$ state.  After $\mathcal{O}(D(T_N))$ rounds, a feedback wave will reach $b$, at which point $b$ will forward the information about its potential leader.  In an additional $\mathcal{O}(D(T_N))$ rounds, the $\mathit{PFC}$ wave completes, at which point the root of $C$ has at least one potential leader (which may or may not be $c$) returned to it.  The root of $C$ will select one returned leader and execute the leader-inform $\mathit{PFC}$ wave.  In $\mathcal{O}(D(T_N))$ rounds, all nodes know the identity of their leader and its cluster.  Cluster $C$'s selected leader $c'$ will be forwarded up the tree during the feedback wave, reaching the root in an additional $\mathcal{O}(D(T_N))$ rounds.
\end{proof}

\begin{lemma}
Let $r_C$ be the root of cluster $C$.  If $r_C$ selects the role of $\mathit{Leader}$, within $\mathcal{O}(D(T_N))$ rounds either $C$ has been paired with a merge partner, or $C$ randomly selects a new role.  Furthermore, all followers of $C$ have been assigned a merge partner.
\label{dlemma:lead_max}
\end{lemma}
\begin{proof}
First, note that $\mathit{PFC}(\mathit{Lead}, \bot)$ requires $\mathcal{O}(D(T_N))$ rounds to complete (Lemma \ref{dlemma:pfc}).  The $\mathit{PFC}(\mathit{ConnectFollowers}, \bot)$ wave also requires $\mathcal{O}(D(T_N))$ rounds.  To see this, notice that the feedback action for this wave cannot advance past a node $b \in C$ until $b$ has no neighbors that are potential followers and all followers are root nodes.  Let $C'$ be a follower cluster that has selected $C$ as a potential leader.  By Lemma \ref{dlemma:follow_to_root}, after $\mathcal{O}(D(T_N))$ rounds, all potential followers of $b$ are either no longer following $b$, or are connected with a root to $b$.

Notice that the \emph{total} wait for all nodes in $C$ is $\mathcal{O}(D(T_N))$, since all nodes in $C$ have a role of $\mathit{ClosedLead}$ and are not assigned any additional potential followers.  Therefore, the feedback wave can be delayed $\mathcal{O}(D(T_N))$ rounds, leading to a total $\mathcal{O}(D(T_N))$ rounds for the $\mathit{PFC}(\mathit{ConnectFollowers}, \bot)$ wave.  When the $\mathit{PFC}(\mathit{ConnectFollowers},\bot)$ wave completes, all followers of $C$ have been assigned a merge partner.  If there were an odd number of followers, $C$ has also been assigned a merge partner, else $C$ will randomly re-select a role.
\end{proof}

\begin{lemma}
Let $C$ be a short follower cluster.  In $\mathcal{O}(D(T_N))$ rounds, either $C$ has selected a leader, or $C$ randomly re-selects a role.
\label{dlemma:short_follow}
\end{lemma}
\begin{proof}
A short follower polls its cluster for a leader at most twice, each requiring $\mathcal{O}(D(T_N))$ rounds (Lemma \ref{dlemma:pfc}).  If a leader is not returned, $C$ will randomly select a new role.  If at least one leader is returned, $C$ selects it.
\end{proof}

\begin{lemma}
Let $C$ be a long follower cluster.  In $\mathcal{O}(D(T_N))$ rounds, either $C$ has selected a leader, or $C$ randomly re-selects a role.
\label{dlemma:long_follow}
\end{lemma}
\begin{proof}
By similar argument to Lemma \ref{dlemma:short_follow}, a long follower polls its cluster at most 12 times, each requiring $\mathcal{O}(D(T_N))$ rounds (Lemma \ref{dlemma:pfc}).  If a leader is not returned, $C$ will randomly select a new role.  If at least one leader is returned, $C$ selects it.
\end{proof}

%

\begin{lemma}
Let $C$ be a follower cluster that has returned a leader after a $\mathit{PFC}$ search wave of Algorithm \ref{dalgo:follow}.  After $\mathcal{O}(D(T_N))$ rounds, $C$ has a merge partner.
\label{dlemma:follow_success_unmatched}
\end{lemma}
\begin{proof}
Let $r_C$ be the root node of $C$.  Node $\mathit{r}_C$ selects a returned leader $C'$ and, in $\mathcal{O}(D(T_N))$ additional rounds, all nodes in $C$ have been informed of leader $C'$ and $\mathit{r}_C$ has an edge to a node $b$ from $C'$.

After $r_C$ has an edge to a node $b \in C'$, $r_C$ waits to be assigned a merge partner.  Suppose $C'$ had the role of leader when selected by $C$.  By Lemma \ref{dlemma:lead_max}, after $\mathcal{O}(D(T_N))$ rounds, $C$ will be assigned a merge partner.

Suppose $C'$ was executing a merge when selected by $C$.  $C$ will be assigned a merge partner when (i) $C'$ finishes its merge, and (iii) $C'$ finishes executing Algorithm \ref{algo:lead}.  If $C'$ was merging, it completes in $\mathcal{O}(D(T_N))$ rounds (Lemma \ref{dlemma:merge}).  The resulting cluster $C''$ will be a leader, and will require $\mathcal{O}(D(T_N))$ rounds before all followers have been assigned a merge partner (Lemma \ref{dlemma:lead_max}).
\end{proof}

Since the initial configuration is set by the adversary, it can be difficult to make probabilistic claims when dealing with the initial configuration.  For instance, the adversary could assign all clusters the role of long follower, in which case the probability that a merge happens over $\mathcal{O}(D(T_N))$ rounds is 0.  Notice, however, that after a short amount of time, regardless of the initial configuration, clusters are guaranteed to have randomly selected their current roles.  Therefore, we ignore the first $\mathcal{O}(D(T_N))$ rounds of execution in the following lemmas.

\begin{definition}
Let $G_0$ be the initial network configuration.  We define $\mathcal{F}_{\Delta}(G_0)$ to be the set of future configurations reached after $\Delta = \mathcal{O}(D(T_N))$ rounds of program execution from the initial configuration.
\end{definition}

\begin{lemma}
Let $C$ be a follower cluster in configuration $G_i \in \mathcal{F}_{\Delta}(G_0)$.  With probability at least $1/2$, $C$ either randomly selects a new role or has found a leader in $\mathcal{O}(D(T_N))$ rounds.
\label{dlemma:follow_to_lead}
\end{lemma}
\begin{proof}
Cluster $C$ must have randomly selected its follower role in $G_i$, as no cluster can be a follower for longer than $\mathcal{O}(D(T_N))$ rounds.  Given that $C$ is a follower, with probability $1/2$, $C$ must have been a short follower, and therefore either $C$ finds a leader or selects a new role after $\mathcal{O}(D(T_N))$ rounds (Lemma \ref{dlemma:short_follow}).
\end{proof}

\begin{lemma}
Consider configuration $G_i \in \mathcal{F}_{\Delta}(G_0)$.  With probability at least $1/4$, every node in cluster $C$ will have been a potential leader for at least one round over the next $\mathcal{O}(D(T_N))$ rounds.
\label{dlemma:lead_probability}
\end{lemma}
\begin{proof}
Consider the possible roles and states of any cluster $C$.  If $C$ is currently participating in a merge or is an $\mathit{OpenLeader}$, then our lemma holds.

Suppose $C$ is a follower in configuration $G_i$.  By Lemma \ref{dlemma:follow_to_lead}, with probability at least $1/2$, $C$ will either find a leader or randomly select another role after $\mathcal{O}(D(T_N))$ rounds.  If $C$ finds a leader, after an additional $\mathcal{O}(D(T_N))$ rounds (Lemma \ref{dlemma:follow_success_unmatched}), $C$ is assigned a merge partner, and after at most $\mathcal{O}(D(T_N))$ additional rounds, all nodes in $C$ are potential leaders.  If $C$ randomly selects another role, with probability $1/2$ $C$ selects the leader role, and all nodes are potential leaders after at most $\mathcal{O}(D(T_N))$ additional rounds.

Suppose a node $b \in C$ is a closed leader ($\mathit{role}_b = \mathit{ClosedLeader}$).  After at most $\mathcal{O}(D(T_N))$ rounds (Lemma \ref{dlemma:lead_max}), either the root of $C$ is assigned a merge partner and $b$ becomes a potential leader after an additional $\mathcal{O}(D(T_N))$ rounds, \emph{or} the root of $C$ is not assigned a merge partner and selects a new role at random.  With probability $1/2$, then, $b$ becomes a potential leader after an additional $\mathcal{O}(D(T_N))$ rounds.
\end{proof}

\begin{lemma}
\label{lemma:merge_partner_probability}
Every cluster $C$ in configuration $G_i \in \mathcal{F}_{\Delta}(G_0)$ has probability at least $1/16$ of being assigned a merge partner over $\mathcal{O}(D(T_N))$ rounds.
\end{lemma}
\begin{proof}
In $\mathcal{O}(D(T_N))$ rounds, if $C$ has not been assigned a merge partner, $C$ will re-select its role.  With probability $1/4$, $C$ will be a long follower and be searching for a leader for an additional $\mathcal{O}(D(T_N))$ rounds.  By Lemma \ref{dlemma:lead_probability}, a neighboring cluster $C'$ has probability at least $1/4$ of being a potential leader during this time.  Therefore, $C$ has probability at least $1/16$ of selecting a leader within $\mathcal{O}(D(T_N))$ rounds, which will result in $C$ being assigned a merge partner after $\mathcal{O}(D(T_N))$ additional rounds (Lemma \ref{dlemma:follow_success_unmatched}).
\end{proof}

\subsection{Merge}
The final component of the algorithm is merging together two matched clusters.  The merge process can be thought of as a series of comparisons between the virtual nodes with the same identifier in the two clusters, with the node with the appropriate host being the ``winner'' in the comparison and remaining in the cluster.  We shall use the spanning tree defined on the target topology to coordinate the merge, beginning with the two cluster roots and moving down level by level.

\subsubsection{Algorithms}
We provide the full algorithms used for merging in our self-stabilizing algorithm: the procedure used in the guest network for replacing a virtual node with another (Figure \ref{algo:replace_node}), and the algorithm for merging the entire cluster (Figure \ref{algo:merge}).

\begin{algorithm}
\begin{tabbing}
........\=....\=....\=....\=....\=....\=....\kill
$\mathit{ReplaceNode}(c, d):$\\
1.\>\textbf{if} $\mathit{partner}_a \neq \mathit{cluster}_b \vee \mathit{partner}_b \neq \mathit{cluster}_a$\\
\>\>$\vee \mathit{rs}_a \neq L \vee \mathit{rs}_b \neq L$ \textbf{then}\\
2.\>\>Reset hosts of nodes $c$ and $d$ (ends the merge process)\\
3.\>\textbf{fi}\\
4.\>\textbf{if} $\mathit{host}_d < \mathit{clusterSucc}_{\mathit{host}_c}$ \textbf{then}\\
5.\>\>$\mathit{clusterSucc}_{\mathit{host}_c} = \mathit{host}_d$\\
6.\>\>$\mathit{LostNodes}_c = \{b : \mathit{host}_b = \mathit{host}_c \wedge b > \mathit{clusterSucc}_{\mathit{host}_c}\}$\\
7.\>\textbf{else if} $\mathit{clusterPred}_{\mathit{host}_c} = \bot \wedge \mathit{host}_d < \mathit{host}_c$ \textbf{then}\\
8.\>\>$\mathit{clusterPred}_{\mathit{host}_c} = \mathit{host}_d$;\\
9.\>\>$\mathit{LostNodes}_c =$\\
\>\>\>$\{b : \mathit{host}_b = \mathit{host}_c \wedge \mathit{host}_d < b < \mathit{host}_c\}$\\
10.\>\textbf{else} \textit{// No successor pointer is updated}\\
11.\>\>Connect cluster children of $c$ to $d$; Delete node $c$\\
12.\>\textbf{fi}\\
13.\>\textbf{for each} $a \in \mathit{LostNodes}_c$ \textbf{do}\\
14.\>\>Copy node $a$ and tree neighbors to $\mathit{host}_d$\\
15.\>\>Delete $a$ from $\mathit{host}_c$'s embedding\\
16.\>\textbf{od}
\end{tabbing}
\caption{The $\mathit{ReplaceNode}$ Procedure}
\label{algo:replace_node}
\end{algorithm}

\begin{algorithm}
\begin{tabbing}
........\=....\=....\=....\=....\=....\=....\kill
\>\textbf{Precondition:} $C$ and $C'$ are merge partners with\\
\>\>connected roots.\\
1.\>$\mathit{root}_C$ ($\mathit{root}_{C'}$) notifies $C$ ($C')$ of \\
\>\>(i) merge partner $C'$ ($C$), and\\
\>\>(ii) value of the shared random sequence.\\
2.\>Remove all \emph{matched edges between $C$ and $C'$}.\\
3.\>$\mathit{ResolveCluster}(\mathit{root}_C, \mathit{root}_{C'})$\\
4.\>Once $\mathit{ResolveCluster}$ completes at leaves,\\
\>\>inform nodes in new cluster $C'' = C \cup C'$ about\\
\>\>new cluster identifier\\
\\
$\mathit{ResolveCluster}(a,b):$ for $a \in T_C(N), b \in T_{C'}(N)$\\
\>\>\textit{// without loss of generality, assume $a \prec b$}\\
1.\>\>$\mathit{ReplaceNode}(a,b)$\\
\>\>\textit{// Node $b$ is now connected to children of $a$.}\\
\>\>\textit{// Let $c_{i,a}$ be the $i$-th child of $a$ in the spanning tree}\\
\>\>\textit{// and $c_{i,b}$ be the $i$-th child of $b$ in the spanning tree.}\\
2.\>\>For each child of $a$ (and $b$), create edge $(c_{i,a}, c_{i,b})$ (that is,\\
\>\>\>connect corresponding children of $a$ and $b$;\\
3.\>\>Concurrently execute $\mathit{ResolveCluster}(c_{i,a}, c_{i,b})$ for all children\\
\end{tabbing}
\caption{The Merge Algorithm}
\label{algo:merge}
\end{algorithm}

\subsubsection{Analysis of Merging}
We present the following lemma concerning the time required to complete a merge between two clusters $C$ and $C'$.

\begin{lemma}
Let $C$ and $C'$ be proper clusters, and let the merge partner of $C$ ($C'$) be $C'$ ($C$).  Assume the root $r_C$ of $C$ and the root $r_{C'}$ of $C'$ are connected.  In $\mathcal{O}(D(T_N))$ rounds, $C$ and $C'$ have merged together into a single proper unmatched clean cluster $C''$, containing exactly $N$ nodes.
\label{dlemma:merge}
\end{lemma}
\begin{proof}
The first step of the merge procedure is to execute the $\mathit{PFC}(\mathit{Prep})$ wave, which requires $\mathcal{O}(D(T_N))$ rounds (Lemma \ref{dlemma:pfc}).  Consider an invocation of the procedure $\mathit{ResolveCluster}(a,b)$.  Let $a$ be from cluster $C$, $b$ be from cluster $C'$, and without loss of generality let $a$ be the virtual node which is to be deleted (that is, the range of the host of $b$ will contain $a$ after the merge).  $\mathit{ReplaceNode}(a,b)$ requires only 1 round, and results in the children of $a$ being connected to node $b$.  In the next round, $b$ will connect its children with the children from $a$, which requires 1 round.  $\mathit{ResolveCluster}$ is then executed concurrently for nodes from level $i+1$.  Therefore, the running time starting from level $i$ is $T(i) = 2 + T(i+1)$.  Since there are $\mathcal{O}(D(T_N))$ levels, we have $T(0) = \sum_{i=0}^{\log N}{2} = \mathcal{O}(D(T_N))$.  After the resolution process reaches the leaves, the final feedback travels up the tree, requiring an additional $\mathcal{O}(D(T_N))$ rounds.
\end{proof}

\begin{lemma}
With probability at least $(1 - N/2^k)$ (where $k = |L|$ and $k \geq \log N$), the algorithm does not disconnect the network.
\end{lemma}
\begin{proof}
First, notice that deletions (of edges and nodes) that occur due to proper clusters $C$ and $C'$ merging do not disconnect the network -- the only edges deleted are those between nodes in $C \cup C'$, and these nodes will form a proper cluster $C''$ after the merge.  The only way in which the network can be disconnected is if the adversary creates an initial configuration such that a node $b$ believes it is either merging, or preparing for a merge, and thus deletes an edge to a node $c$.  Notice that for any edge $(b,c)$ to be deleted, both $b$ and $c$ must have the same value for their random sequence, and this value must match the shared random string $L$.  While the adversary can enforce the first condition, they are unable to guarantee the second.  Instead, for any particular pair of nodes $a$ and $b$, the adversary has probability $1/2^k$ of setting the random sequences of $a$ and $b$ to match $L$.  An adversary can have up to $N/2$ different ``guesses'' in any initial configuration.  Therefore, the probability that the network is disconnected is at most $N/2^{k+1}$ (for $k \geq \log N$).
\end{proof}

\subsection{Final Time Complexity}
Putting together the pieces for each component above yields the following result.

\begin{theorem}
The algorithm of Berns~\cite{berns_avatar_15} defines a self-stabilizing overlay network for $\textsc{Avatar}_T$, for some full graph family target topology $T$, with convergence time of $\mathcal{O}(D(T_N) \cdot \log N)$ in expectation.
\end{theorem}
\begin{proof}
By Lemma~\ref{dlemma:reset_to_clean}, every node is part of a proper cluster in $\mathcal{O}(D(T_N))$ rounds.  By Lemma \ref{lemma:merge_partner_probability}, every cluster has probability at least $1/16$ of being assigned a merge partner over a span of $\mathcal{O}(D(T_N))$ rounds, meaning in $\mathcal{O}(D(T_N))$ rounds in expectation, every cluster will be assigned a merge partner.  By Lemma \ref{dlemma:merge}, in $\mathcal{O}(D(T_N))$ rounds, two matched clusters will complete their merge and form a single cluster.  This implies that the number of clusters is decreased by a constant fraction in $\mathcal{O}(D(T_N))$ rounds.  This decrease by a constant fraction will happen $\mathcal{O}(\log N)$ times before a single cluster remains, and our theorem holds.
\end{proof}

\section{Updated Analysis of Berns~\cite{berns_avatar_15}: Degree Expansion}
In this section, we update the analysis of Berns~\cite{berns_avatar_15} to show the degree expansion of their algorithm is an expected $\mathcal{O}(\Delta_\textsc{A}(T_N) \cdot \log N)$, where $\Delta_\textsc{A}(T_N)$ is the maximum degree of embedding of an $N$ node topology $T$.

To begin, we define the set of actions a node may execute that can increase its degree.

\begin{definition}
Let a \emph{degree-increasing action} of a virtual node $b$ be any action that adds a node $c$ to the neighborhood of a node $b' \in N_b$ such that $b'$ is not hosted by $\mathit{host}_b$.  Specifically, the degree-increasing actions are:
\begin{enumerate}
\item (Matching for Leader Clusters): edges added from the connecting and forwarding of followers during the $\mathit{PFC}(\mathit{ConnectFollowers},\bot)$ wave of Algorithm \ref{algo:lead}
\item (Matching for Follower Clusters): forwarding an edge to a potential leader after the root has begun a search for an open leader in Algorithm \ref{dalgo:follow}
\item (Merge: Replace Node): virtual node transfer actions during Algorithm \ref{algo:merge}
\item (Merge: Resolve Cluster): connecting children from each cluster with the same identifiers as the merge process propagates through the tree in Algorithm \ref{algo:merge}.
\end{enumerate}
\end{definition}

Given the degree-increasing actions, we can bound the degree expansion with the following two lemmas.

\begin{lemma}
Let $u$ be a real node in configuration $G_i$.  The maximum number of real nodes $u$ will have added to its neighborhood before all virtual nodes hosted by $u$ are members of a proper clean cluster is $\mathcal{O}(\Delta_\textsc{A}(T_N))$.
\label{dlemma:reset_increase}
\end{lemma}
\begin{proof}
First, notice that a real node will detect a reset fault if it hosts two virtual nodes $b$ and $b'$ such that $b$ and $b'$ are executing different algorithms -- for example, if $b$ is merging while $b'$ is executing a selection for leaders, host $u$ will reset and not forward any neighbors.  Furthermore, note that all degree-increasing actions happen as part of PFC waves on the spanning tree.  Given these two observations, consider the three degree-increasing actions:
\begin{itemize}
    \item \emph{Matching for Leader Clusters:} Note that connecting and forwarding of followers happens only during the feedback portion of a PFC wave.  This requires $\mathcal{O}(D(T_N))$ rounds to complete once.  By Lemma \ref{dlemma:reset_to_clean}, every node will be a member of a proper clean cluster in $\mathcal{O}(D(T_N))$ rounds, meaning the leader selection will happen at most once before the node is in a proper clean cluster.  Each time a leader is selected the node's degree may grow by $\mathcal{O}(\Delta_\textsc{A}(T_N))$, and our lemma holds.
    
    \item \emph{Matching for Follower Clusters:} As with the matching process for leaders, followers only forward an edge to a potential leader as part of a PFC wave.  Furthermore, each child only forwards a single edge to its parent.  Since a PFC wave requires $\mathcal{O}(D(T_N))$ rounds to complete, and nodes are part of a proper clean cluster in $\mathcal{O}(D(T_N))$ rounds, at most a constant number of waves can occur (1, to be exact).  Each node has $\mathcal{O}(\Delta_\textsc{A}(T_N))$ children from other hosts in the cluster, so this is the bound on the degree increase before reaching a proper clean cluster.
    
    \item \emph{Replace Node:} If a real node believes it is participating in a merge, it will receive at most $\mathcal{O}(MaxExt_T(N))$ additional edges from each range it takes over.  As it can only believe it is merging with one other cluster, it can take over at most 2 such ranges before either (i) the merge completes successfully and all virtual nodes $u$ hosts are a member of a proper clean unmatched cluster, or (ii) some other node detects a faulty configuration, executes a reset fault, and this reset fault spreads to node $u$.
    
    \item \emph{Resolve Cluster:} As with other degree-increasing actions, connecting children from two clusters during a merge happens in a systematic fashion level by level.  Since this level-by-level process takes $\mathcal{O}(D(T_N))$ rounds, it can occur at most a constant number of times before nodes are in a proper clean cluster.  Note that degree of a real node can increase from this only if the parent of the connected child is hosted by a different node than the child.  By definition, there are at most $\mathcal{O}(\Delta_\textsc{A}(T_N))$ such parents, and therefore this is a limit of the degree increase during this step before the nodes are part of a proper clean cluster.
\end{itemize}
\end{proof}

Now that we have bounded the degree growth during the initial stages, we can consider how the degrees may grow as part of the intended operation (after all nodes are members of proper clusters).

\begin{lemma}
Let $u$ be a real node such that all virtual nodes hosted by $u$ are members of a proper clean unmatched cluster in configuration $G_i$.  Let $u$'s degree in $G_i$ be $\Delta_u$.  Node $u$'s degree will be $\mathcal{O}(\Delta_u + \Delta_\textsc{A}(T_N) + \Delta_\textsc{A}(T_N) \cdot T(\mathit{lead}))$ until the algorithm terminates, where $T(\mathit{lead})$ is the number of times the virtual nodes hosted by $u$ participate in the leader selection procedure of Algorithm \ref{algo:lead}.
\label{dlemma:degree_clean}
\end{lemma}
\begin{proof}
We consider the four degree-increasing actions that a proper clean unmatched cluster $T$ from configuration $G_i$ may execute.  Consider first the \emph{matching for follower clusters} from Algorithm \ref{dalgo:follow}.  The degree can increase only from node $b$ adding the neighbor $\mathit{leader}$ from cluster $T'$ to the neighborhood of $\mathit{parent}_b$.  Furthermore, this degree increase of one is temporary -- a node deletes an edge to $\mathit{leader}$ after forwarding it, and once the root has the edge, it eventually becomes part of a merge, and either the root of $T$ or the root of $T'$ is deleted.

Next, consider the \emph{matching for leader clusters} in Algorithm \ref{algo:lead}.  During the $\mathit{PFC}(\mathit{ConnectFollowers},\bot)$ wave, a virtual node $b$ may receive at most a single neighbor from each child, and after an additional round will retain at most half of these edges.  Since a real node $u$ can host at most $\mathcal{O}(\Delta_\textsc{A}(T_N))$ nodes with children from another host, the degree increase each time a node $u$ participates in the matching procedure for leaders is $\mathcal{O}(\Delta_\textsc{A}(T_N))$.

Next, consider the actions of connecting children from the two clusters from the \emph{ResolveCluster} procedure.  Notice the degree of a real node can increase from this step only if the parent of the connected children is hosted by a different node than the child (otherwise the edge is already present in the host network).  By definition, for any real node there are at most $\Delta_\textsc{A}(T_N)$ parents in the spanning tree that are hosted by another node, therefore the connecting of children can increase the degree by at most $\Delta_\textsc{A}(T_N)$.  Notice that after an additional round, these extra edges are deleted as the connected children resolve using the \emph{ResolveNode} procedure.

Finally, consider the degree increase from the \emph{ResolveNode} procedure.  Assume $b$ and $b'$ are nodes in level $i$ in $C$ and $C'$ (respectively), and suppose $b$ and $b'$ resolve.  Without loss of generality, assume $b'$ is the node which will be deleted (that is, the host of $b'$ is losing part of its range).  Notice that the degree of the $host_b$ will increase by at most $\Delta_\textsc{A}(T_N)$ when $host_b$ ``takes over'' the lost range of $host_{b'}$, as this is the highest external degree of $b'$ in cluster $C'$.  A real node can take over a lost range exactly once per merge, however, as there is only one node in the other cluster which will update its successor to be $b$.  Therefore, the degree increase caused from the taking over of a node's range is at most $\Delta_\textsc{A}(T_N)$ per merge.  Notice that, unlike the degree increase from the selection procedure for leaders, the degree increase from merges is not additive -- the largest node $u$'s degree can be as the result of intra-cluster edges is $\Delta_\textsc{A}(T_N)$, regardless of how many merges $u$ participates in.  Therefore, a node's degree may grow up to $\Delta_\textsc{A}(T_N)$ during a merge, but it will never exceed this value.
\end{proof}

Putting the above lemmas together, along with the results on the convergence time, leads us to the following theorem.

\begin{theorem}
The algorithm of Berns~\cite{berns_avatar_15} defines a self-stabilizing overlay network for $\textsc{Avatar}_T$, for some full graph family target topology $T$, with degree expansion of $\mathcal{O}(\Delta_\textsc{A}(T_N) \cdot \log N)$ in expectation.
\end{theorem}
\begin{proof}
By Lemma \ref{dlemma:reset_increase}, we know the maximum degree expansion before a node is a member of a proper cluster is $\mathcal{O}(\Delta_\textsc{A}(T_N))$.  By Lemma \ref{dlemma:degree_clean}, once a node is a member of the proper cluster its degree expansion will be $\mathcal{O}(\Delta_\textsc{A}(T_N) \cdot T(lead))$, where $T(lead)$ is the number of times $u$'s virtual nodes participate in the leader selection procedure.  Since this procedure takes $\mathcal{O}(D(T_N))$ rounds to complete, and the entire algorithm completes in expected $\mathcal{O}(D(T_N) \cdot \log N)$ rounds, $u$'s virtual nodes participate in the leader selection procedure $\mathcal{O}(\log N)$ times, giving us our result.
\end{proof}
\end{document}